\documentclass[conference]{IEEEtran}
\IEEEoverridecommandlockouts
\usepackage{cite}
\usepackage{amsmath,amssymb,amsfonts}
\usepackage{algorithmic}
\usepackage{graphicx}
\usepackage{textcomp}
\usepackage{xcolor}

\usepackage{lipsum}
\usepackage{mathtools}
\usepackage{cuted}

\usepackage[nodisplayskipstretch]{setspace}
\usepackage{bm}
\usepackage{subfigure}
\usepackage{amsthm} % This define \begin {proof}
\def\BibTeX{{\rm B\kern-.05em{\sc i\kern-.025em b}\kern-.08em
		T\kern-.1667em\lower.7ex\hbox{E}\kern-.125emX}}
\newcommand{\mv}[1]{\mbox{\boldmath{$ #1 $}}}
\newtheorem{proposition}{Proposition}
\newtheorem{remark}{Remark}

\makeatletter
\def\endthebibliography{%
	\def\@noitemerr{\@latex@warning{Empty `thebibliography' environment}}%
	\endlist
}
\makeatother
\begin{document}
	
	\title{MIMO Integrated Sensing and Communication with Extended Target: CRB-Rate Tradeoff \\
%		{\footnotesize \textsuperscript{*}Note: Sub-titles are not captured in Xplore and
%			should not be used}
%		\thanks{Identify applicable funding agency here. If none, delete this.}
	}
	\author{\IEEEauthorblockN{Haocheng~Hua\IEEEauthorrefmark{1}, Xianxin~Song\IEEEauthorrefmark{1}, Yuan Fang\IEEEauthorrefmark{1},
			Tony Xiao~Han\IEEEauthorrefmark{2}, and
			Jie~Xu\IEEEauthorrefmark{1}
		}
		\IEEEauthorblockA{\IEEEauthorrefmark{1}School of Science and Engineering and Future Network of Intelligence Institute,\\ The Chinese University of Hong Kong (Shenzhen), Shenzhen, China}
		\IEEEauthorblockA{\IEEEauthorrefmark{2}The 2012 Lab, Huawei, Shenzhen, China \\
			Email: haochenghua@link.cuhk.edu.cn,~xianxinsong@link.cuhk.edu.cn,~fangyuan@cuhk.edu.cn,\\~tony.hanxiao@huawei.com,~xujie@cuhk.edu.cn}
			\thanks{Jie Xu is the corresponding author.}
	}
	
	\maketitle
	
\begin{abstract}
	This paper studies a multiple-input multiple-output (MIMO) integrated sensing and communication (ISAC) system, in which a multi-antenna base station (BS) sends unified wireless signals to estimate an extended target and communicate with a multi-antenna communication user (CU) at the same time. We investigate the fundamental tradeoff between the estimation Cram\'er-Rao bound (CRB) for sensing and the data rate for communication, by characterizing the Pareto boundary of the achievable CRB-rate (C-R) region. Towards this end, we formulate a new MIMO rate maximization problem by optimizing the transmit covariance matrix at the BS, subject to a new form of maximum CRB constraint together with a maximum transmit power constraint. We derive the optimal transmit covariance solution in a semi-closed form, by first implementing the singular-value decomposition (SVD) to diagonalize the communication channel and then properly allocating the transmit power over these subchannels for communication and other orthogonal subchannels (if any) for dedicated sensing. It is shown that the optimal transmit covariance is of full rank, which unifies the conventional rate maximization design with water-filling power allocation and the CRB minimization design with isotropic transmission. Numerical results are provided to validate the performance achieved by our proposed optimal design, in comparison with other benchmark schemes.
	%		 show the C-R region
	%		 and the achievable rate with varying transmit power under given CRB constraint
	%		 achieved by the optimal design, as compared to other benchmarking schemes such as the time sharing and the power split based designs. The corresponding optimal power allocation are also shown in comparison with water-filling and equal power allocation. Finally, we examine the achievable rate with varying Signal-to-noise ratio (SNR) under fixed CRB constraint, which shows that the optimal design can be approximated by the two power split based designs under different scenarios.
\end{abstract}

%	\begin{IEEEkeywords}
%		Integrated sensing and communication (ISAC), multiple antennas, Cram¨¦r-Rao bound (CRB), Data rate.
%	\end{IEEEkeywords}

\section{Introduction}
\label{sec:intro}
Recently, integrated sensing and communication (ISAC) has been recognized as a candidate technology towards future sixth-generation (6G) cellular networks to enable various environment-aware intelligent applications (see, e.g., \cite{liu2022integrated} and the references therein), in which wireless signals and cellular infrastructures are reused for both sensing and communication functionalities. Motivated by the great success of multi-antenna or multiple-input multiple-output (MIMO) techniques in wireless communications \cite{telatar1999capacity,Tse2005book} and radar sensing \cite{li2007mimo,stoica2007probing,haimovich2007mimo} independently, MIMO ISAC has attracted particular research interests, in which the multiple antennas can be exploited to provide spatial multiplexing and diversity gains to increase the communication rate and reliability \cite{telatar1999capacity,Tse2005book}, as well as the waveform/spatial diversity gains to enhance the sensing accuracy and resolution \cite{li2007mimo,stoica2007probing,haimovich2007mimo,song2022intelligent}. In the literature, there have been some prior works \cite{liu2018toward,zhang2018multibeam, Eldar2020joint,xu2021rate,hua2021optimal,liu2021cramer,song2022joint} studying multi-antenna ISAC designs to optimize the sensing and communication performance. However, these existing works mainly focused on practical waveform and beamforming approaches that are generally suboptimal in achieving the performance limits for sensing and/or communication.

%Despite the progress, these prior works considered either practical transmit beamforming designs \cite{zhang2018multibeam} \textcolor{blue}{or adopting other design objectives for communication such as multi-user interference (MUI) \cite{liu2018toward}, signal-to-noise-plus-interference ratio (SINR) \cite{Eldar2020joint, hua2021optimal, liu2021cramer} and weighted sum rate (WSR) combined with sensing objective \cite{xu2021rate}}

In ISAC systems, it is essential to understand the performance tradeoffs between radar sensing and communication from detection/estimation and information theories \cite{liu2022survey}. This not only helps reveal the fundamental performance limits, but also guides practical ISAC system designs. On one hand, the Cram\'er-Rao bound (CRB) serves as a fundamental performance metric for radar estimation, which defines the variance lower bound by any unbiased estimators. On the other hand, the channel capacity acts as a fundamental performance metric for wireless communication,  which captures the rate upper bound by any practical  modulation and coding schemes. Therefore, how to characterize the fundamental CRB-rate (C-R) tradeoff for MIMO ISAC systems is becoming an important problem to be tackled. To our best knowledge, only one recent work \cite{xiong2022flowing} studied the so-called C-R region for a point-to-point MIMO ISAC system with one sensing target, which is defined as the set containing all C-R pairs that can be simultaneously achieved by sensing and communication. However, \cite{xiong2022flowing} only obtained two boundary points on the C-R region, namely the CRB-minimization and rate-maximization points, at which the CRB is minimized and the rate is maximized, respectively. Unfortunately, \cite{xiong2022flowing} failed to characterize the whole boundary of the C-R region, especially the boundary points between the above two. This thus motivates our work to fill in such a research gap.

In particular, this paper considers a point-to-point MIMO ISAC system with an extended target, in which a multi-antenna base station (BS) sends unified wireless signals to estimate an extended target from the echo and communicate with a multi-antenna communication user (CU) at the same time. We aim to reveal the fundamental C-R tradeoff of this system, by characterizing the Pareto boundary of the C-R region. The main results are listed as follows.
\begin{itemize}
	\item First, to characterize the C-R-region boundary between the CRB-minimization and rate-maximization points, we formulate a new CRB-constrained rate maximization problem, in which the data rate for MIMO communication is maximized by optimizing the transmit covariance matrix, subject to a new form of maximum CRB constraint and the maximum transmit power constraint.
	
	%We characterize the C-R region explicitly by outlining the corresponding pareto boundary and find out the closed-form expression for the two corner points under different conditions. To achieve each point on the curve between the two corners, we
	%derive the semi-closed-form solution of the optimal transmit covariance by first implementing the singular-value decomposition (SVD) to diagonalize the communication channel and then properly allocating the transmit power over orthogonal communication and sensing subchannels.	
	\item Next, we derive the optimal transmit covariance solution to the CRB-constrained rate maximization problem in a semi-closed form. Towards this end, we first implement the singular-value decomposition (SVD) to diagonalize the communication channel, and accordingly transform the transmit covariance optimization problem into an equivalent power allocation problem over these subchannels for communication and other orthogonal subchannels (if any) for dedicated sensing. Then, we obtain the optimal power allocation solution by the Lagrange duality method. It is shown that the optimal transmit covariance is of full rank, which unifies the conventional rate maximization design with water-filling power allocation and the CRB minimization design with isotropic transmission.
	
	%
	%rigorously prove that the formulated problem can be reformulated equivalently as a subchannel power allocation design problem by . By introducing , we derive the semi-closed-form solution of the optimal power allocation and the optimal transmit covariance matrix can thus be accordingly obtained.
	%	It is revealed analytically that the optimal power allocation unifies the sensing-oriented equal power allocation and the communication-oriented water-filling power allocation.
	%	In particular, the power allocation towards each subchannel is a monotonically increasing function of the subchannel gain, which is similar as conventional water-filling. However, all the subchannels will be allocated with some power regardless of their subchannel gain to meet the sensing requirement, which is quite similar to equal power allocation favoured by sensing task.
	\item Finally, we present numerical results to evaluate the C-R-region boundary achieved by the optimal transmit covariance by considering two cases with rank-deficient and full-column-rank communication channels, respectively, as compared with other benchmark schemes.
	%
	%compared to those by other benchmark schemes.
	%
	%with heuristic power allocation designs
	%
	%and show the corresponding optimal power allocation. Besides the optimal design strategy, other sub-optimal designs have been proposed to compare with the optimal one and are shown to approach the optimal strategy under different scenarios.
	% some insights here later...
\end{itemize}

%Will be more efficient to add after other parts are finalized and complete...

{\it Notations:} Boldface letters refer to vectors (lower  case) or matrices (upper case). For a square matrix $\mv{S}$, ${\operatorname{tr}}(\mv{S})$  denotes its trace, and $\mv{S}\succeq \mv{0}$ means that $\mv{S}$ is positive semidefinite. For an arbitrary-size matrix $\mv{M}$, $\text{det}(\bm{M})$, ${\text{rank}}(\mv{M})$, $\mv{M}^H$, and $\mv{M}^T$ denote its determinant, rank, conjugate transpose, and transpose, respectively. $\otimes$ and $\circ$ denote the Kronecker product and the Hadamard product, respectively.
%The distribution of a circularly symmetric complex Gaussian (CSCG) random vector with mean vector $\mv{x}$ and covariance matrix $\mv{\Sigma}$ is denoted by $\mathcal{CN}(\mv{x,\Sigma})$; and $\sim$ stands for ``distributed as''.
$\mathbb{R}^{x\times y}$ and $\mathbb{C}^{x\times y}$ denote the spaces of real and complex matrices, respectively. {${\mathbb{E}}\{\cdot\}$} denotes the statistical expectation. $\|\mv{x}\|$ denotes the Euclidean norm of a complex vector $\mv{x}$. $|z|$ and $z^*$ denote the magnitude and the conjugate of a complex number $z$, respectively. For a real number $x$, $\left(x\right)^+ = \max(x,0)$. $\operatorname{diag}(x_1,...,x_n)$ denotes a diagonal matrix with diagonal elements $x_1,...,x_n$.

\section{System Model}\label{sec:system1}

We consider a MIMO ISAC system, in which a BS communicates with a CU and simultaneously estimates an extended target, as shown in Figs. \ref{fig:system_model}(a) and \ref{fig:system_model}(b) with monostatic and bistatic sensing, respectively. There are $M > 1$ transmit antennas at the BS transmitter (Tx), $N_s$ receive antennas at the BS receiver (Rx), and $N_c > 1$ antennas at the CU.
%\footnote{the design is applicable to the bistatic case when the receiver knows the transmitted message. (in e.g. C-RAN scenario )}.
%The BS sends a unified signal to perform both communication and sense at the same time, which is denoted by $x(n)$ at symbol $n$.
%We consider the monostatic setup, in which the echo wave of the target is received in the same BS with a separated ULA of $N_s$ antennas at its receiver side, as shown in Fig. \ref{fig:system_model}.
%Without loss of generality, We assume $M < N_s$ throughout the discussion, as the receive antenna number should be larger than the transmit antenna to avoid information loss of the sensed target. % this is the conclusion from Liu fan work, if you cannot figure it the real reason and it is not important, delete it. Do not simply copy or just add a reference.
\begin{figure}[htb]
	\centering
	\setlength{\abovecaptionskip}{-1mm}
	\setlength{\belowcaptionskip}{-1mm}
	\includegraphics[width=3.2in]{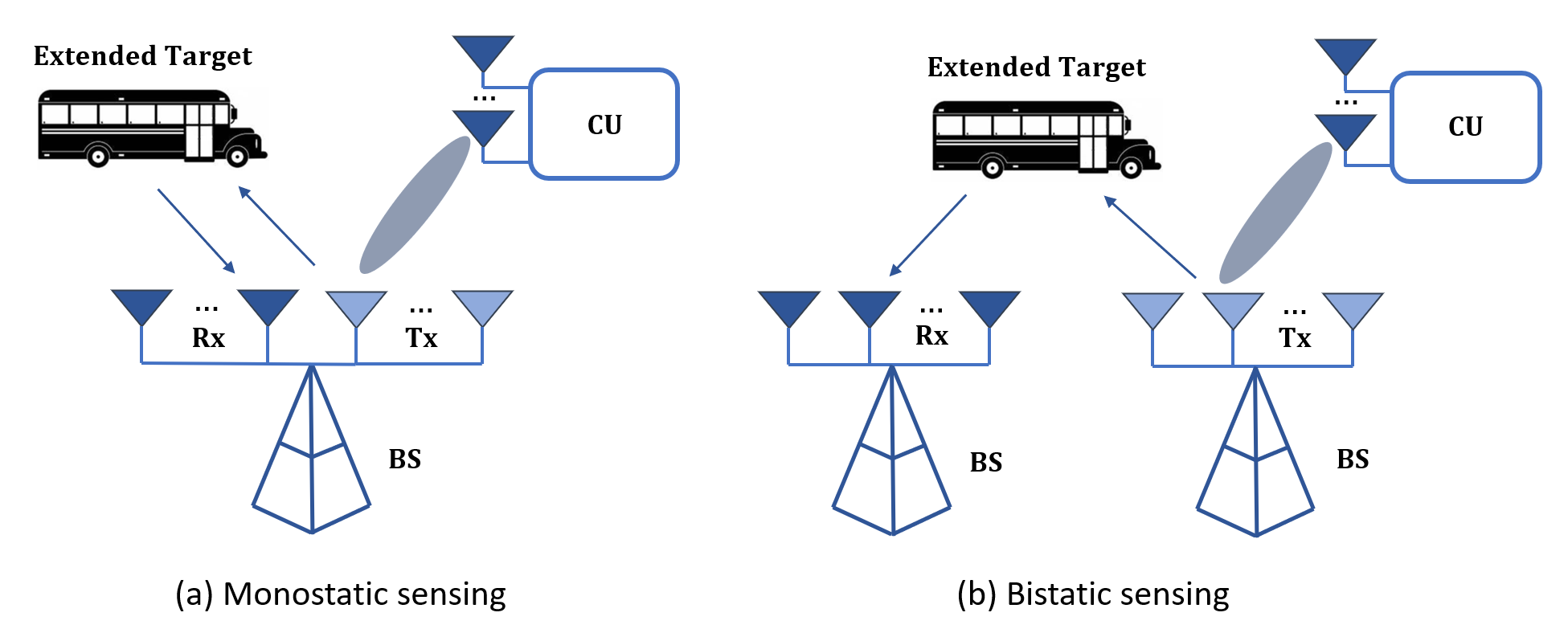}
	\caption{Illustration of the MIMO ISAC system. } \label{fig:system_model}
	%	\textbf{[Check the template, the caption should start with Fig.]}
\end{figure}

Let $\bm{x}(n)$ denote the  unified transmit signal at symbol $n$ for both communication and sensing. It is assumed that $\bm{x}(n)$ is a circular symmetric complex Gaussian (CSCG) random vector with zero mean and covariance matrix $\bm{Q} = \mathbb{E} \{\bm{x}(n) \bm{x}^H(n)\} \succeq \bm{0}$, i.e., $\bm{x}(n) \sim \mathcal{CN}(\bm{0}, \bm{Q})$. Let $P$ denote the transmit power budget at the BS-Tx. Then we have the power constraint as
\begin{align}\label{eqn:power:constraint}
	\operatorname{tr} (\bm{Q}) = \mathbb{E}\{\|\bm{x}(n)\|^2\} \leq P.
\end{align}

% with $1 \leq n \leq L$, where $L > M$ is the length of the radar pulse/communication frame. In general, the communication may last longer time than sensing due to the infinite block length. It is assumed that the transmitted signal is known to the sensing receiver, but unknown to the CU receiver.
First, we consider the point-to-point MIMO communication. Let $\bm{H}_c \in \mathbb{C}^{N_c \times M}$ denote the channel matrix from the BS-Tx to the CU, whose rank is denoted by $r = \text{rank}(\bm{H}_c) \leq \min(N_c,M)$. The received signal by the CU at symbol $n$ is
\begin{align}
	\bm{y}_c(n) = \bm{H}_c \bm{x}(n) + \bm{z}_c(n),
\end{align}
where $\bm{z}_c(n)$ denotes the noise at the CU receiver that is a CSCG random vector with zero  mean and covariance $\sigma_c^2 \bm{I}_{N_c}$, i.e., $\bm{z}_c(n) \sim \mathcal{CN}(\bm{0},\sigma_c^2 \bm{I}_{N_c})$. In this case, the achievable rate (in bps/Hz) of the MIMO channel with $\bm{Q}$ is
\begin{align}\label{eqn:Rate}
	R(\bm{Q}) = \log_2 \det \left(\bm{I}_{N_c} + \frac{1}{\sigma_c^2} \bm{H}_c \bm{Q} \bm{H}_c^H \right).
\end{align}
It is assumed that the channel matrix $\bm{H}_c$ is perfectly known at the BS-Tx, such that it can design $\bm{Q}$ based on $\bm{H}_c$ to optimize the achievable rate $R(\bm{Q})$.

Next, we consider the MIMO radar sensing over a particular coherent processing interval (CPI) with $L > M$ symbols or radar pulses. Let $\mathcal{L} = \{1,\ldots, L\}$ denote the set of symbols in the CPI, and $\bm{X} = \left[\bm{x}(1),...,\bm{x}(L)\right] \in \mathbb{C}^{M \times L}$ denote the transmitted signals over the CPI. Suppose that $\bm{H}_s \in  C^{N_s \times M}$ denotes the target response matrix from the BS-Tx to the target to the BS-Rx. Accordingly, the received echo signal $\bm{Y}_s \in \mathbb{C}^{N_s \times L}$ at the BS-Rx is
\begin{align}
	\bm{Y}_s = \bm{H}_s \bm{X} + \bm{Z}_s,
\end{align}
where $\bm{Z}_s \in \mathbb{C}^{N_s \times L}$ denotes the noise matrix at the BS-Rx, with each element being a CSCG random variable with zero mean and variance $\sigma_s^2$.
%an AWGN matrix with the variance of each entry being $\sigma_s^2$, and $\bm{H}_s \in \mathbb{C}^{N_s \times M}$ is the sensing channel matrix. The goal of the BS is to recover $\bm{H}_s$ given $\bm{Y}_s$ and transmitted unified signal matrix $\bm{X}$. Notice that when $L$ is sufficiently long\footnote{When $L$ is finite, then the covariance matrix might be different from $\tilde{\bm{Q}}$, in which practical sequence design is necessary, which is left for future work.}, then the covariance matrix $\bm{Q}$ can be approximated by $\bm{X}$ via $\tilde{\bm{Q}} = \frac{1}{L} \bm{X} \bm{X}^H$.
In particular, we consider the case with an extended target, which is modelled as the combination of a large number of $K$ distributed point-like scatterers. Suppose that the target is located at a fixed location during the CPI. In this case, $\bm{H}_s$ is expressed as \cite{liu2021cramer}
\begin{align}
	\bm{H}_s=\sum_{k=1}^{K} \alpha_{k} \bm{b}\left(\phi_{k}\right) \bm{a}^{T}\left(\theta_{k}\right),
\end{align}
where $\alpha_{k}$ denotes the reflection coefficient of the $k$-th scatterer, $\theta_{k}$ and $\phi_{k}$ denote its associated angle of departure (AoD) and angle of arrival (AoA) at the BS, and $\bm{a}\left(\theta_{k}\right)$ and $\bm{b}\left(\phi_{k}\right)$ denote the corresponding transmit and receive steering vectors, respectively. The objective of sensing is to estimate the target response matrix $\bm{H}_s$, which contains $M N_s$ complex parameters. %Here, we use the CRB matrix as the performance metric for estimating the parameters in $\bm{H}_s$, which serves as a lower bound on the variance of any unbiased estimators.
In this case, the CRB matrix for estimating $\bm{H}_s$ is given by \cite{liu2021cramer}
\begin{align}\label{eqn:CRB_matrix}
	\overline{\bold{CRB}}(\bm{Q}) = \bm{J}(\bm{Q})^{-1}
\end{align}
which is a complex matrix with dimension $M N_s \times M N_s$, with the $(M(i-1)+j)$-th diagonal element representing the lower bound of variance for unbiasedly estimating the $(i,j)$-element of $\bm{H}_s$, $i \in \{1, \ldots, N_s\}, j\in\{1, \ldots, M\}$. In \eqref{eqn:CRB_matrix}, $\bm{J}(\bm{Q})$ is the Fisher information matrix given by \cite{liu2021cramer}
\begin{align}
	\bm{J}(\bm{Q})=\frac{1}{\sigma_{s}^{2}} \bm{X}^{*} \bm{X}^{T} \otimes \bm{I}_{N_{s}} \approx  \frac{L}{\sigma_{s}^{2}} \bm{Q}^T \otimes \bm{I}_{N_{s}},
\end{align}
where $\frac{1}{L} \bm{X} \bm{X}^H$ is approximated as $\bm{Q}$ by assuming that $L$ is sufficiently large \cite{liu2021cramer}.
% Notice that each diagonal element of the CRB matrix represents the corresponding MSE lower bound of any estimator. To optimize those sensing performance bounds in an efficient and tractable way, the typical approach is to choose a surrogate variable such that optimizing the surrogate variable will ensure that the sensing performance bounds can also be optimized \cite{li2007range}.
% To simplify the illustration, we name this surrogate variable as ??Surrogate CRB (SCRB)??.
Based on the CRB matrix $\overline{\bold{CRB}}(\bm{Q})$ in \eqref{eqn:CRB_matrix}, we use its trace as the sensing performance metric for estimating $\bm{H}_s$ \cite{kay1993fundamentals,ben2009constrained}, i.e.,
\begin{align}\label{eqn:CRB}
	\text{CRB}(\bm{Q}) = \operatorname{tr}(\overline{\bold{CRB}}(\bm{Q})) = \operatorname{tr}(\bm{J}(\bm{Q})^{-1}) = \frac{\sigma_{s}^2 N_s}{L} \operatorname{tr}(\bm{Q}^{-1}).
\end{align}
The BS-Tx can design $\bm{Q}$ to optimize $\text{CRB}(\bm{Q})$ for estimation.

%We aim to maximize the capacity at the communication user while making sure that the CRB of the target is less than a threshold and the sum power constraints at the BS are satisfied. \\

%The problem under point target and extended target can be formulated as follows respectively.
%The problem under extended target can be formulated as follows.

\section{C-R Region Characterization}\label{section-region}

This section characterizes the C-R region to reveal the fundamental tradeoff between the data rate $R(\bm{Q})$ in \eqref{eqn:Rate} for communication and the estimation CRB $\text{CRB}(\bm{Q})$ in \eqref{eqn:CRB} for sensing. To start with, we define the C-R region, which is a set containing all C-R pairs that can be simultaneously achievable by the ISAC system under the given transmit power constraint.  Mathematically, the C-R region with power budget $P$ is defined as
\begin{align}\label{equ:C-R-region}
	\nonumber
	&\mathcal{C}_{C-R}(P)  \triangleq \{(\bar{\Gamma},\bar{R}): \bar{\Gamma}\ge\frac{\sigma_{s}^2 N_s}{L} \operatorname{tr}(\bm{Q}^{-1}), \\
	&~~\bar{ R} \leq \log_2 \det (\bm{I}_{N_c} + \frac{1}{\sigma_c^2} \bm{H}_c \bm{Q} \bm{H}_c^H), \operatorname{tr}(\bm{Q}) \leq P, \bm{Q} \succeq \bm{0}\}.
\end{align}
In this case, revealing the optimal tradeoff between communication rate $R(\bm{Q})$ and estimation CRB $\text{CRB}(\bm{Q})$ corresponds to finding the Pareto boundary of the C-R region $\mathcal{C}_{C-R}(P)$ in (\ref{equ:C-R-region}). Towards this end, we first introduce two boundary points corresponding to rate maximization and CRB minimization, respectively.
%In order to characterize the C-R region, we need to maximize the capacity and minimize the CRB.
%(An example of C-R region has been given in Fig. \ref{fig:C-R-region}. Its setup and the C-R region under other conditions will be elaborated in more details in Section \ref{Section_results}.)

%we need to first get two corner points $P_C$ and $P_S$ shown in Fig. \ref{fig:C-R-region}.

First, we maximize the achievable rate $R(\bm{Q})$ by optimizing the transmit covariance $\bm{Q}$, i.e.,
\begin{align}\label{equ:P_R}
	\max _{\bm{Q}\succeq \bm{0}}  \text{ } \log_2 \det (\bm{I}_{N_c} + \frac{1}{\sigma_c^2} \bm{H}_c \bm{Q} \bm{H}_c^H) \quad \text {s.t. }  \operatorname{tr}(\bm{Q}) \leq P.
\end{align}
Based on the SVD, we have $\bm{H}_c = \mv{U}_c \mv{\Sigma} \mv{V}_c^H$, where $\mv{U}_c \in \mathbb{C}^{N_c \times N_c}$ and $\mv{V}_c \in \mathbb{C}^{M \times M}$ with $\mv{U}_c^H \mv{U}_c =\mv{U}_c \mv{U}^H_c = \bm{I}_{N_c}$ and $\mv{V}_c^H \mv{V}_c = \mv{V}_c \mv{V}_c^H = \bm{I}_{M}$, and $\mv{\Sigma} \in \mathbb{C}^{N_c \times M}$ is an all-zero matrix except the first $r$ diagonal elements being the $r$ non-zero singular values $\lambda_1 \geq \lambda_2 \geq ... \geq \lambda_r > 0$. It has been well established in \cite{telatar1999capacity} that the optimal solution to the rate maximization problem (\ref{equ:P_R}) is given by $\bm{Q}_c^* = \mv{V}_c \bm{\Lambda} \mv{V}_c^H$, where $\bm{\Lambda} = \operatorname{diag}(p_{c,1}^*,...,p_{c,r}^*,0,...,0)$ denotes the water-filling power allocation matrix with its first $r$ diagonal elements given by
\begin{align}\label{eq:WF}
	p_{c,i}^* = \left(\nu - \frac{\sigma_c^2}{\lambda_i^2}\right)^+, \forall i \in \{1,\ldots, r\}.
\end{align}
In (\ref{eq:WF}), $\nu$ is the water level that can be obtained based on $\sum_{i=1}^r p_{c,i}^* = P$. At the obtained $\bm{Q}_c^*$, let  $R_{\text{max}} = R(\bm{Q}_c^*) =  \sum_{i=1}^{r} \log_2 (1+\frac{\lambda_i^2 p_{c,i}^*}{\sigma_c^2})$ and $\text{CRB}_C = \text{CRB}(\bm{Q}_c^*)$ denote the correspondingly achieved data rate and estimation CRB, respectively. As a result, we obtain the rate-maximization boundary point of the C-R region as $(\text{CRB}_C, R_{\text{max}})$.
%
%the corresponding
%
%
%The following remark discusses different conditions based on which the corresponding CRB at the obtained $\bm{Q}_c^*$, denoted by $\text{CRB}_C$, can be obtained.
\begin{remark} \label{remark:finite_SCRB}
	\emph{It is worth noting that if the rate-maximization transmit covariance $\bm{Q}_c^*$ is rank-deficient (i.e., $\text{rank}(\bm{Q}_c^*) < M$), then it follows from \eqref{eqn:CRB} that $\text{CRB}_C \rightarrow \infty$. This means that $\bm{H}_s$ is not estimable in this case due to the lack of degrees of freedom (DoF). Accordingly, the rate-maximization boundary point becomes $(\infty, R_{\text{max}})$. It can be verified that this case happens when the communication channel matrix $\bm{H}_c$ is rank-deficient (i.e., {$\text{rank}(\bm{H}_c) = r < M$}) or the transmit power is small (i.e., $P \le {P_0 \triangleq } \sum_{i=1}^{M-1}  (\frac{\sigma_c^2}{\lambda_M^2}-\frac{\sigma_c^2}{\lambda_i^2})$).}

	%
	%	\emph{When $\text{rank}(\bm{H}_c) = r = M$ and $P > P_0 = \sum_{i=1}^{M-1} (\frac{\sigma_c^2}{\lambda_M^2}-\frac{\sigma_c^2}{\lambda_i^2})$, $\text{CRB}_{C}$ is finite and explicitly given by
	%		\begin{align}\label{equ:Gamma-R-finite}
	%		\text{CRB}_{C} = \frac{\sigma_{s}^2 N_s}{L} \sum_{i=1}^{M} 1/(\frac{P}{M}+\frac{\sum_{i=1}^{M} \frac{\sigma_c^2}{\lambda_i^2}}{M} - \frac{\sigma_c^2}{\lambda_i^2}).
	%		\end{align}
	%		The above observation is easy to see since to have finite $\text{CRB}_C$, we need to make sure that the power $P$ is large enough and $\text{rank}(\bm{H}_c) = r = M$ such that the resultant power allocation obtained by standard water-filling will assign power to all the $M$ subchannels. We can then accordingly obtain (\ref{equ:Gamma-R-finite}).
	%		Otherwise, when $ r < M$ or $ r = M$ but $P \leq  P_0 $,
	%		%	under the following two scenarios,
	%		%	\begin{itemize}
	%		%		\item [(a)] $\text{rank}(\bm{H}_c) = r < M$.
	%		%		\item [(b)] $\text{rank}(\bm{H}_c) = r = M$ but $P \leq  P_0 $.
	%		%	\end{itemize}
	%		$\text{CRB}_{C}$ is infinite, which means that $\bm{H}_s$ is not estimable due to the lack of Degree-of-freedom (DoF).}
\end{remark}
%and $\text{CRB}_{C}$ can be accordingly calculated as $\text{CRB}_{C} = \frac{\sigma_{s}^2 N_s}{L} \sum_{i=1}^{M} \frac{1}{p_{ci}^*}$. This corresponds to corner point $ (\text{CRB}_C, R_{\text{max}})$.
%\noindent Combining the above discussion, the first corner point $(\text{CRB}_C, R_{\text{max}})$ can then be obtained.

Next, we minimize the CRB $\text{CRB}(\bm{Q})$ by optimizing the transmit covariance $\bm{Q}$, i.e.,
\begin{align}\label{equ:P_CRB}
	%		\text{(P.C)} \text{ } \min _{\bm{Q}_s} & \text{ } \frac{\sigma_{s}^2 N_s}{L} \operatorname{tr}(\bm{Q}_s^{-1}) \\
	\min _{\bm{Q}\succeq \bm{0}}  \text{ } \frac{\sigma_{s}^2 N_s}{L} \operatorname{tr}(\bm{Q}^{-1}) \quad \text {s.t. }  \operatorname{tr}(\bm{Q}) \leq P.
\end{align}
By checking the Karush-Kuhn-Tucker (KKT) conditions, the optimal solution to problem \eqref{equ:P_CRB} is obtained as $\bm{Q}_s^* = \frac{P}{M} \bm{I}_M$ \cite{liu2021cramer}. As a result, the correspondingly achieved minimum CRB and data rate become $\text{CRB}_{\text{min}} = \frac{\sigma_{s}^2 N_s M^2}{P L}$ and $R_{S} = \sum_{i=1}^r \log_2 \left( 1+ \frac{\lambda_i^2 P}{\sigma_c^2 M} \right)$, respectively. Therefore, we obtain the CRB-minimization boundary point as $(\text{CRB}_{\text{min}},R_S)$.

%\begin{subequations} \label{equ:P_CRB_min_achieved}
%\begin{align}
%	M^2 & \leq \left(\sum_{i=1}^M \frac{1}{p_i}\right) \left( \sum_{i=1}^M p_i \right) = \frac{\text{SCRB}(\bm{H}_s) L}{\sigma_{s}^2 N_s} \left( \sum_{i=1}^M p_i \right) \\
%	 & \leq \frac{\text{SCRB}(\bm{H}_s) L}{\sigma_{s}^2 N_s} P
%\end{align}
%\end{subequations}
%where the equality is achieved when $p_1 = p_2 = ... = p_M = \frac{P}{M}$. Thus, the optimal solution to (P.C) is $\frac{P}{M} \bm{I}$. The corresponding $\Gamma_{\text{min}} = \frac{\sigma_{s}^2 N_s M^2}{P L}$ with the resultant capacity $R_{C} = \sum_{i=1}^M \log_2 \left( 1+ \frac{\lambda_i^2 P}{\sigma_c^2 M} \right)$, where $\lambda_1 \geq \lambda_2 \geq ... \geq \lambda_M \geq 0$ are the singular values of $\bm{H}_c$.
%In other words, the first corner point is $(R_{C}, \Gamma_{\text{min}})$.

%Notice that there exist two cases for the corresponding achieved SCRB, which are discussed in the following two propositions.

%Notice that in order for the achieved $\text{CRB}_C$ to be finite, we need to ensure that $\bm{Q}_c^*$ is of full rank. This corresponds to the following scenario.

Based on the obtained two boundary points $(\text{CRB}_C, R_{\text{max}})$ and $(\text{CRB}_{\text{min}},R_S)$, it now remains to find the remaining Pareto boundary points between them for characterizing the whole C-R region. To find each boundary point, we propose to maximize the achievable rate $R(\bm{Q})$ by optimizing $\bm{Q}$, subject to the maximum CRB constraint $\text{CRB}(\bm{Q}) \leq {\Gamma}$ and the transmit power constraint in \eqref{eqn:power:constraint}, where the CRB threshold $\Gamma$ is set as a constant such that $\text{CRB}_{\text{min}}\le {\Gamma} \le \text{CRB}_{C}$. By defining $\tilde{\Gamma} \triangleq \frac{L \Gamma}{\sigma_s^2 N_s}$, the CRB-constrained rate maximization problem is formulated as
%Our goal is to design the transmitted signal covariance matrix to maximize the capacity at the communication user while making sure that the CRB of the extended target is less than a threshold and the sum power constraints at the BS are satisfied. Mathematically, the sensing-constrained capacity maximization problem is formulated as
\begin{subequations} \label{equ:P_1}
	\begin{align}
		\label{equ:Obj_P_1}
		\text{(P1)}: \text{ } \max _{\bm{Q}\succeq \bm{0}} & \text{ } \log_2 \det (\bm{I}_{N_c} + \frac{1}{\sigma_c^2} \bm{H}_c \bm{Q} \bm{H}_c^H) \\
		\label{equ:CRLB_P_1}
		\text { s.t. } & \operatorname{tr}(\bm{Q}^{-1}) \leq \tilde{\Gamma} \\
		\label{equ:Power_P_1}
		&\operatorname{tr}(\bm{Q}) \leq P.
	\end{align}
\end{subequations}
%where  for the ease of illustration in the later discussion. Note that by exhausting  the CRB threshold $\Gamma$ between $\text{CRB}_{\text{min}}$ and $\text{CRB}_{C}$, we can find all boundary points between $(\text{CRB}_C, R_{\text{max}})$ and $(\text{CRB}_{\text{min}},R_S)$.
Note that problem (P1) is convex and thus can be solved optimally based on standard convex optimization techniques \cite{boyd2004convex}. To gain insights, we derive its optimal solution in a semi-closed form in the next section.

\begin{remark} \label{remark:P_infinite_discussion}
	\emph{It is worth discussing the C-R region in the special case when the communication channel is of full column rank ($r = M$) and the transmit power is sufficiently large (i.e., $P \to \infty$). In this case, it is easy to show that $\bm{Q}_c^* = \bm{Q}_s^* = \frac{P}{M} \bm{I}_M$ and the two boundary points become identical (i.e., $(\text{CRB}_{\text{min}},R_S) = (\text{CRB}_C, R_{\text{max}})$). {As a result}, the C-R region can be obtained as a box without solving problem (P1), which is denoted by $\mathcal{C}_{C-R}(P) = \{(\bar{\Gamma},\bar{R}): \bar{\Gamma} \ge \text{CRB}_{\text{min}},\bar{R} \le R_{\text{max}}\}$.
		% can easily check that  $\lim\limits_{P \to \infty} \frac{R_{\text{max}}}{R_S} = 1, \lim\limits_{P \to \infty} \text{CRB}_{\text{min}} = 0 $ for $r \leq M$ while $\lim\limits_{P \to \infty} \text{CRB}_{C} = 0$ for $r = M$ and $\text{CRB}_{C} = \infty$ for $r<M$, as before. Under both $r<M$ and $r=M$ cases, the C-R region will converge to ??unit step-like?? function.
		%\begin{figure}[htb]
		%	\centering
		%	%\epsfxsize=1\linewidth
		%	\setlength{\abovecaptionskip}{-1mm}
		%	\setlength{\belowcaptionskip}{-4mm}
		%	\includegraphics[width=2.3in]{figures/C_R_region_P_inf.png}
		%	%	\includegraphics[width=3.3in]{figures/globecom2022_fig_1_region_teqM.jpg}
		%	\caption{C-R region when $P \to \infty$.} \label{fig:C-R-region-P-inf}
		%	%	\textbf{[Check the template, the caption should start with Fig.]}
		%\end{figure}	
	}
	
\end{remark}

\section{Optimal Solution to Problem (P1)}
This section presents the optimal solution to (P1). First, recall that the SVD of $\bm{H}_c$ is $\bm{H}_c = \mv{U}_c \mv{\Sigma} \mv{V}_c^H$. By defining
\begin{align}\label{equ:Q_transform}
	\tilde{\bm{Q}} \triangleq \mv{V}_c^H \mv{Q} \mv{V}_c~{\text{or}}~{\bm{Q}} \triangleq \mv{V}_c \tilde{\mv{Q}} \mv{V}_c^H,
\end{align}
(P1) can be equivalently reformulated as
\begin{subequations}\label{equ:P1_equivalent1}
	\begin{align}
		\label{equ:P1_1_obj}
		\text{(P1.1):} \text{ } \max _{\tilde{\bm{Q}}\succeq \bm{0}}  \text{ }& \log_2 \det (\mv{I}_{M} + \frac{1}{\sigma_c^2}  \mv{\Sigma}^2 \tilde{\bm{Q}} ) \\
		\label{equ:P1_1_CRB}
		\text { s.t. } & \operatorname{tr}( \tilde{\bm{Q}}^{-1}) \leq  \tilde{\Gamma} \\
		\label{equ:P1_1_power}
		&\operatorname{tr}(\tilde{\bm{Q}}) \leq P.
		%		&\mathbf{Q} \succeq \mv{0}
	\end{align}
\end{subequations}
where $\mv{\Sigma}^2 \triangleq \mv{\Sigma}^H \mv{\Sigma} = \operatorname{diag}(\lambda_{1}^2,...,\lambda_{r}^2,0,...,0) \in \mathbb{R}^{M \times M}$. Here, (\ref{equ:P1_1_obj}) follows from \eqref{equ:Obj_P_1} based on the fact that $\det (\bm{I}_{N_c} + \frac{1}{\sigma_c^2} \bm{H}_c \bm{Q} \bm{H}_c^H) = \det (\bm{I}_{N_c} + \frac{1}{\sigma_c^2} \bm{U}_c \mv{\Sigma} \mv{V}_c^H \mv{V}_c \tilde{\bm{Q}} \mv{V}_c^H \mv{V}_c  \mv{\Sigma}^H \bm{U}_c^H)  = \det (\mv{I}_{M} + \frac{1}{\sigma_c^2}  \mv{\Sigma}^2 \tilde{\bm{Q}} )$. Furthermore, (\ref{equ:P1_1_CRB}) and (\ref{equ:P1_1_power}) follow from \eqref{equ:CRLB_P_1}
and	\eqref{equ:Power_P_1}, due to the facts that $\operatorname{tr}({\bm{Q}}^{-1}) = \operatorname{tr}((\mv{V}_c \tilde{\mv{Q}} \mv{V}_c^H)^{-1}) = \operatorname{tr}(\mv{V}_c {\tilde{\mv{Q}}}^{-1} \mv{V}_c^{H}) = \operatorname{tr}(\mv{V}_c^H \mv{V}_c {\tilde{\mv{Q}}}^{-1} ) = \operatorname{tr}({\tilde{\mv{Q}}}^{-1} )$ and $\operatorname{tr}({\bm{Q}}) = \operatorname{tr}( \mv{V}_c \tilde{\mv{Q}} \mv{V}_c^H) = \operatorname{tr}(\mv{V}_c^H \mv{V}_c \tilde{\mv{Q}} ) = \operatorname{tr}(\tilde{\mv{Q}})$, respectively.
%(\ref{equ:P1_1_obj}) follows due to $\det (\bm{I}_{N_c} + \frac{1}{\sigma_c^2} \bm{H}_c \bm{Q} \bm{H}_c^H) = \det (\bm{I}_{N_c} + \frac{1}{\sigma_c^2} \bm{U}_c \mv{\Sigma} \tilde{\bm{Q}} \mv{\Sigma}^H \bm{U}_c^H) = \det (\mv{I}_{M} + \frac{1}{\sigma_c^2}  \mv{\Sigma}^H \mv{\Sigma} \tilde{\bm{Q}} )$, (\ref{equ:P1_1_CRB}) holds due to $\operatorname{tr}(\bm{Q}^{-1}) = \operatorname{tr}(\bm{V}_c \tilde{\bm{Q}}^{-1} \bm{V}_c^H ) = \operatorname{tr}(\tilde{\bm{Q}}^{-1} \bm{V}_c^H \bm{V}_c) = \operatorname{tr}(\tilde{\bm{Q}}^{-1})$, and (\ref{equ:P1_1_power}) follows due to $\operatorname{tr}(\bm{Q}) = \operatorname{tr}(\bm{V}_c \tilde{\bm{Q}} \bm{V}_c^H ) = \operatorname{tr}(\tilde{\bm{Q}} \bm{V}_c^H \bm{V}_c) = \operatorname{tr}(\tilde{\bm{Q}})$.
Next, we have the following proposition.
%\begin{lemma}\label{lemma:trace_ineq}
%	\emph{ Let $\mv{A} \in \mathbb{C}^{n \times n}$ be positive definite hermitian. Then
%		$$
%		\operatorname{tr} \{(\mv{A} \circ \mv{I})^{-1}\} \leq \operatorname{tr} \{ \mv{A}^{-1}\}
%		$$
%		with equality iff $\mv{A}$ is diagonal. }
%	\begin{proof}
%		Please refer to \cite{ohno2004capacity}.
%	\end{proof}
%\end{lemma}
%Next, we have the following proposition, which shows that the optimal solution to (P1.1) is a diagonal positive definite matrix.
\begin{proposition}\label{Pro:diagonal_optimal}
	\emph{The optimal solution to problem (P1.1) is a diagonal matrix with positive diagonal elements, i.e., $\tilde{\bm{Q}} = \operatorname{diag}(p_1,p_2,...,p_M)$, where $p_i > 0, \forall i \in \{1,\ldots,M\}$.}
\end{proposition}
\begin{proof}
	First, it is evident that $\tilde{\bm{Q}}$ must be positive definite (or $\tilde{\bm{Q}}\succ \bm{0}$) in order for the maximum CRB constraint in \eqref{equ:P1_1_CRB} to hold. Next, suppose that the optimal solution is a positive definite matrix $\tilde{\bm{Q}}^*$ that is not diagonal, and we construct an alternative solution $\tilde{\bm{Q}}^{**} = \tilde{\bm{Q}}^* \circ \mv{I}$, which is a diagonal matrix whose diagonal elements are identical to $\tilde{\bm{Q}}^*$. Then, we have 
	\begin{align}
		\label{eqn:1}
		\det (\mv{I}_{M} + \frac{1}{\sigma_c^2}  \mv{\Sigma}^2 \tilde{\bm{Q}}^* )
		\le & \det (\mv{I}_{M} + \frac{1}{\sigma_c^2}  \mv{\Sigma}^2 \tilde{\bm{Q}}^{**} ),\\
		%	= & \sum_{i=1}^r \log_2 (1+\frac{\lambda_i^2 [\tilde{\bm{Q}}^*]_{i,i}}{\sigma_c^2})
		\operatorname{tr} (\tilde{\bm{Q}}^{**}) = &\operatorname{tr} (\tilde{\bm{Q}}^*) \leq P\label{eqn:2},\\
		\operatorname{tr} \{(\tilde{\bm{Q}}^{**})^{-1}\} \leq &\operatorname{tr} \{ (\tilde{\bm{Q}}^*)^{-1}\} \leq \tilde{\Gamma},\label{eqn:3}
	\end{align}
	where \eqref{eqn:1} follows from the Hadamard inequality \cite{horn2012matrix} and the inequality in \eqref{eqn:3} follows from \cite[Lemma 1]{ohno2004capacity}. By combining \eqref{eqn:1}, \eqref{eqn:2}, and \eqref{eqn:3}, it is clear that $\tilde{\bm{Q}}^{**}$ is also feasible for problem (P1.1) and achieves a no lower objective value than that by $\tilde{\bm{Q}}^{*}$. This contradicts the presumption that the non-diagonal matrix $\tilde{\bm{Q}}^*$ is optimal. This thus verifies that the optimal solution of $\tilde{\bm{Q}}$ to (P1.1) must be diagonal, i.e., $\tilde{\bm{Q}} = \operatorname{diag}(p_1,p_2,...,p_M)$, where $p_i > 0, \forall i \in \{1,...,M\}$.
	% Together with the fact that $\tilde{\bm{Q}}$ must be positive definite, we have $p_i > 0, \forall i \in \{1,\ldots,M\}$. This thus completes the proof.
	%Thus, $\tilde{\bm{Q}}^* \circ \mv{I}$ is indeed another feasible solution and it increases the objective function value, this inherently results in a contradiction to the assumption that $\tilde{\bm{Q}}^*$ is the optimal solution. Thus, the optimal solution to (P1.1) must be a diagonal matrix with all positive diagonal elements, i.e., $\tilde{\bm{Q}} = \operatorname{diag}(\bm{p}) = \operatorname{diag}(p_1,...,p_M)$, where $p_i>0, \forall i$.
\end{proof}
Based on Proposition \ref{Pro:diagonal_optimal}, problem (P1.1) is equivalently reformulated as
%we now can further transform (P1.1) into the following format by letting $\tilde{\bm{Q}} = \operatorname{diag}(\bm{p}) = \operatorname{diag}(p_1,p_2,...,p_M)$, where $p_i > 0, \forall i$ can be viewed as the power allocation in the $i$-th subchannel.
\begin{subequations}
	\begin{align}\label{equ:P_1P}
		\text{(P1.2):} \text{ } \max _{\{p_i \ge 0\}} & \sum_{i=1}^r \log_2 (1+\frac{\lambda_i^2 p_i}{\sigma_c^2}) \\
		\label{equ:P_1P_CRB}
		\text { s.t. } & \sum_{i=1}^M \frac{1}{p_i} \leq  \tilde{\Gamma} \\
		\label{equ:P_1P_Power}
		& \sum_{i=1}^M p_i \leq P.
		%		&\mathbf{Q} \succeq \mv{0}
	\end{align}
\end{subequations}

Then, we find the optimal solution to (P1.2) as follows.
\begin{proposition}\label{pro:prime_dual_relationship}
	\emph{
		For problem (P1.2), \textcolor{black}{the optimal solution of $\{p_i^{\text{opt}}\}_{i=1}^r$ must satisfy the following conditions:
		\begin{align}
			\label{equ:Lag_zero1_in_pro}
			\frac{1}{\text{ln}2}  (\frac{1}{1+\frac{\lambda_i^2 p_i^{\text{opt}}}{\sigma_c^2}})& \frac{\lambda_i^2}{\sigma_c^2} + (\frac{\mu^{\text{opt}}}{(p_i^{\text{opt}})^2}) - v^{\text{opt}} = 0,  \forall i \in \{1,\ldots, r\},		\end{align}
	or, equivalently,} we have
			\begin{align}
	p_i^{\text{opt}} =& -t_1 + \sqrt[3]{-t_2+\sqrt{t_2^2+t_3^3}} +  \sqrt[3]{-t_2-\sqrt{t_2^2+t_3^3}}, \nonumber \\ &\quad \quad \quad \quad \quad \quad\quad\quad \quad \quad \forall i \in \{1,\ldots, r\},	\label{equ:general_exp_optimal}
%				& \quad \quad\quad\quad\quad \quad \quad \quad\quad\quad\quad \quad \forall i \in \{1,\ldots, r\}, \label{equ:general_exp} \\
			\end{align}
			where
			$$
			t_1 = \frac{b_i}{3a}, t_2 = \frac{27a^2d_i-9ab_ic+2b_i^3}{54a^3}, t_3 = \frac{3ac-b_i^2}{9a^2},
			$$
			with $a = v^{\text{opt}}, b_i = v^{\text{opt}} \frac{\sigma_c^2}{\lambda_i^2} - \frac{1}{\text{ln2}}, c = -\mu^{\text{opt}}$, and $d_i = -\mu^{\text{opt}} \frac{\sigma_c^2}{\lambda_i^2}$. 
			Furthermore, the optimal solution of $\{p_i^{\text{opt}}\}_{i=r+1}^M$ is given by
			\begin{align}
						p_i^{\text{opt}} =& \sqrt{\frac{\mu^{\text{opt}}}{v^{\text{opt}}}}, \quad \forall i \in \{r+1,\ldots, M\}. \label{equ:vanish_exp}
			\end{align}
			Here, $\mu^{\text{opt}}$ and $v^{\text{opt}}$ are the optimal dual variables associated with the constraint in (\ref{equ:P_1P_CRB}) and (\ref{equ:P_1P_Power}), respectively\footnote{The optimal dual solution of $\mu^{\text{opt}}$ and $v^{\text{opt}}$ can be obtained by solving the dual problem of (P1.2) via subgradient-based methods, for which the details can be found in Appendix \ref{Proof:prime_dual_relationship}.}.}\end{proposition}
\begin{proof}
As problem (P1.2) is convex and satisfies the slater's condition, the strong duality holds between (P1.2) and its Lagrange dual problem \cite{boyd2004convex}. The optimal solution to (P1.2) can be found by using the Lagrange duality method. Please see Appendix \ref{Proof:prime_dual_relationship} for details.		%See Appendix \ref{Proof:diagonal_optimal}.
%	See Appendix \ref{Proof:prime_dual_relationship}.
\end{proof}

Finally, by combining (\ref{equ:Q_transform}) with Propositions \ref{Pro:diagonal_optimal} and \ref{pro:prime_dual_relationship}, the optimal solution to (P1) is obtained as
\begin{align}\label{equ:Optimal_sol_P1}
	\bm{Q}^{\text{opt}} = \bm{V}_c \tilde{\bm{Q}}^{\text{opt}}\bm{V}_c^H,
\end{align}
where $\tilde{\bm{Q}}^{\text{opt}} = \operatorname{diag}(p_1^{\text{opt}},\ldots, p_M^{\text{opt}})$ with $\{p_i^{\text{opt}}\}$ given in Proposition \ref{pro:prime_dual_relationship}.
%Having obtained the optimal $\bm{p}^{\text{opt}}$, the optimal solution to (P1) can the be constructed as $$.
%Based on what we have already obtained so far, we can now reveal a set of engineering insights behind the considered problem (P1).

\subsection{Optimal Solution Structures}

To gain more insights, this subsection discusses the structure of the optimal transmit covariance solution $\bm{Q}^{\text{opt}}$ in \eqref{equ:Optimal_sol_P1}. In particular, we express $\bm{V}_c$ as $\bm{V}_c = [\bar{\bm{V}}_{c}, \hat{\bm{V}}_c]$, where $\bar{\bm{V}}_{c} \in \mathbb{C}^{M\times r}$ consists of the first $r$ right singular vectors of the communication channel $\bm{H}_c$, and $\hat{\bm{V}}_{c}\in \mathbb{C}^{M\times (M-r)}$ consists of the other $M-r$ ones that span the null space of $\bm{H}_c$. In this case, the optimal transmit covariance solution $\bm{Q}^{\text{opt}}$ in \eqref{equ:Optimal_sol_P1} can be equivalently rewritten as
\begin{align}\label{equ:Optimal_sol_P1:eqv}
	\bm{Q}^{\text{opt}} = \bar{\bm{V}}_c \bar{\bm{Q}}^{\text{opt}}\bar{\bm{V}}_c^H + \hat{\bm{V}}_c \hat{\bm{Q}}^{\text{opt}}\hat{\bm{V}}_c^H,
\end{align}
where $\bar{\bm{Q}}^{\text{opt}} = \operatorname{diag}(p_1^{\text{opt}},\ldots, p_r^{\text{opt}})$ and $\hat{\bm{Q}}^{\text{opt}} = \operatorname{diag}(p_{{r+1}}^{\text{opt}},\ldots, p_M^{\text{opt}})$.

It is interesting to observe from \eqref{equ:Optimal_sol_P1:eqv} that the transmit covariance $\bm{Q}^{\text{opt}}$ is separated into two parts, including  $\bar{\bm{V}}_c \bar{\bm{Q}}^{\text{opt}}\bar{\bm{V}}_c^H$ lying in the range of $\bm{H}_c^H$ for both communication and sensing and $\hat{\bm{V}}_c \hat{\bm{Q}}^{\text{opt}}\hat{\bm{V}}_c^H$ lying in the null space of $\bm{H}_c$ for dedicated sensing only. As the right singular matrix $\bm{V}_c = [\bar{\bm{V}}_{c}, \hat{\bm{V}}_c]$ actually diagonalizes the communication channel $\bm{H}_c$ into $r$ parallel subchannels, it is clear that $\bar{\bm{Q}}^{\text{opt}} = \operatorname{diag}(p_1^{\text{opt}},\ldots, p_r^{\text{opt}})$ corresponds to the optimized power allocation over the $r$ parallel communication subchannels, and $\hat{\bm{Q}}^{\text{opt}} = \operatorname{diag}(p_{{r+1}}^{\text{opt}},\ldots, p_M^{\text{opt}})$ corresponds to that over the other orthogonal $M-r$ dedicated sensing subchannels.

%When $\bm{H}_c$ is rank deficient, i.e., $r<M$, we name the first $r$ subchannels as ISAC subchannels while the other $M-r$ ones as dedicated sensing channels. The following lemma shows the power allocation in the rank deficiency case.
\begin{proposition}\label{lemma:Lemma_order_p}
	\emph{The optimal power allocation satisfies that $p_1^{\text{opt}} \geq ... \geq p_r^{\text{opt}} \geq p_{r+1}^{\text{opt}} = ... = p_{M}^{\text{opt}} > 0 $ .}
\end{proposition}
	\begin{proof}
See Appendix \ref{Proof:lemma_power_allocation}.		%See Appendix \ref{Proof:diagonal_optimal}.
%		See Appendix \ref{Proof:lemma_power_allocation}.
	\end{proof}
Proposition \ref{lemma:Lemma_order_p} shows that the power allocations over communication subchannels (i.e., $\{p_i^{\text{opt}}\}_{i=1}^r$) are monotonically increasing with respect to the subchannel gains $\{\lambda_i^2\}_{i=1}^r$, which is similar as the conventional water-filling power allocation in \eqref{eq:WF} for rate maximization. By contrast, the power allocations (i.e., $\{p_i^{\text{opt}}\}_{i=r+1}^M$) are constant over dedicated sensing subchannels, similarly as that for CRB minimization (see \eqref{equ:P_CRB}). As a result, the optimal power allocation for ISAC in Proposition \ref{pro:prime_dual_relationship} unifies the above two conventional power allocations for independent communication and sensing, respectively.

%However, all the subchannels will be allocated with some power regardless of their subchannel gains and the power allocation to sensing subchannels will all be equal, which is quite similar to equal power allocation favoured by sensing task. %The following remark further discusses the case when $\text{rank}(\bm{H}_c) = r = M$, i.e., $\bm{H}_c$ is full column rank.
%\begin{remark} \label{remark:Mgeqr_strategy}
%	\emph{When $\text{rank}(\bm{H}_c) = r = M$, the power allocation  will follow the similar strategy as specified in Lemma \ref{lemma:Lemma_order_p}, i.e., $p_1^{\text{opt}} \geq ... \geq p_M^{\text{opt}} > 0 $ and can be proved similarly.}
%\end{remark}

Finally, it is also interesting to discuss the optimal power allocation in the special case with $P \to \infty$.
\begin{proposition}\label{pro:P_infinite}
	\emph{When $P \to \infty$, the optimal power allocation for problem (P1.2) is given by
 \begin{align}\label{equ:P_infinite}
		p_i^{\text{opt}} = \begin{cases} \frac{1}{r} (P-\frac{(M-r)^2}{\tilde{\Gamma}}), & 1 \leq i \leq r \\ \frac{M-r}{\tilde{\Gamma}}, & r+1 \leq i \leq M \end{cases},
		\end{align}
in which the transmit power is split into two parts over communication and dedicated sensing subchannels, with equal power allocation within each part.
% is intuitively to see that given $\tilde{\Gamma}$, the sensing subchannels will be allocated with equal power to meet the CRB constraint according to Lemma \ref{lemma:Lemma_order_p}, and the rest of the power will be equally allocated to the ISAC channels to increase the data rate. Mathematically,		
}
\end{proposition}
\begin{proof}
See Appendix \ref{Proof:Pro_P_infinite}.	%See Appendix \ref{Proof:diagonal_optimal}.
%		See Appendix \ref{Proof:lemma_power_allocation}.
\end{proof}

\section{Numerical Results} \label{Section_results}

This section presents numerical results to validate the C-R region performance of the presented optimal transmit covariance, as compared to the following benchmark schemes.
\begin{itemize}
	\item {\bf Time switching}: The BS time switches the two transmit covariances $\bm{Q}_c^*$ and $\bm{Q}_s^*$ for rate maximization and CRB minimization, respectively. This design is only applicable when $\bm{Q}_c^*$ is of full rank, since otherwise $\text{CRB}_C = \text{CRB}(\bm{Q}_c^*) \rightarrow \infty$ follows (see Remark \ref{remark:finite_SCRB}).
%When $r = M$ and $P$ is sufficiently large, which corresponds to the opposite case in Remark \ref{remark:finite_SCRB} where $\bm{Q}_c^*$ is full rank, Time-sharing design simply switches between $\bm{Q}_c^*$ and $\bm{Q}_s^*$ in different time slots.
	\item {\bf Power splitting with equal power allocation (EP)}: Similarly as in \eqref{equ:Optimal_sol_P1}, the BS sets the transmit covariance as $\bm{Q}^{\text{EP}} = \bm{V}_c \tilde{\bm{Q}}^{\text{EP}}\bm{V}_c^H$, in which $\tilde{\bm{Q}}^{\text{EP}} = \operatorname{diag}(p_1^{\text{EP}},\ldots, p_M^{\text{EP}})$ denotes the power allocation. The BS splits the transmit power $P$ into two parts, $\beta P$ for the $r$ communication subchannels and $(1-\beta)P$ for the $M-r$ sensing subchannels, with $0\le \beta\le 1$ denoting the power splitting factor that is a parameter to be optimized. Following the equal power allocation, we have $p_1^{\text{EP}} = \ldots = p_r^{\text{EP}} =  \frac{\beta P}{r}$ and $p_{r+1}^{\text{EP}} = \ldots = p_M^{\text{EP}} =  \frac{(1-\beta) P}{M-r}$. Notice that if $r = M$, then we set $\beta = 1$.
%When $ r < M$, this design splits $P$ into two parts, $\beta P$ for $M-r$ sensing subchannels and $(1-\beta)P$ for $r$ communication subchannels, where $0 \leq \beta \leq 1$.
%	In particular, all $\beta P$ are allocated equally towards all the $M-r$ sensing subchannels while all $(1-\beta)P$ are allocated equally towards all the other $r$ subchannels.
	%The resulting $\bm{p}^{\text{EP}}$ can be expressed as
%	\begin{align}
%		\nonumber
%		\bm{p}^{\text{EP}} = [ \underbrace{\frac{(1-\beta)P}{r},...,\frac{(1-\beta)P}{r}}_{r}, \underbrace{\frac{\beta P}{M-r},...,\frac{\beta P}{M-r}}_{M-r} ]
%	\end{align}
%	and it is substituded into (P1.2) and we maximize $\beta$ to obtain the CRB-constrained maximum rate.
	\item {\bf Power splitting with strongest eigenmode transmission (SEM)}: The BS sets $\bm{Q}^{\text{SEM}} = \bm{V}_c \tilde{\bm{Q}}^{\text{SEM}}\bm{V}_c^H$, in which $\tilde{\bm{Q}}^{\text{SEM}} = \operatorname{diag}(p_1^{\text{SEM}},\ldots, p_M^{\text{SEM}})$. The BS splits the transmit power $P$ into two parts, $\beta P$ for the the dominant communication subchannel and $(1-\beta)P$ for the remaining $M-1$ subchannels, with $0\le \beta\le 1$ to be optimized. In this case, we have $p_1^{\text{SEM}} = \beta P$ and $p_{2}^{\text{SEM}} = \ldots = p_M^{\text{EP}} =  \frac{(1-\beta) P}{M-1}$.
%
%
%This design can be used for all the three considered conditions, in which $(1-\beta)P$ is now all allocated towards the strongest subchannel. $\beta P$ are then equally allocated towards all the other $M-1$ subchannels.
%	The resulting $\bm{p}^{\text{SEM}}$ can be expressed as
%	\begin{align}
%		\nonumber
%		\bm{p}^{\text{SEM}} = [ (1-\beta)P, \underbrace{\frac{\beta P}{M-1},...,\frac{\beta P}{M-1}}_{M-1} ]
%	\end{align}
\end{itemize}

%depict the  obtained by the proposed optimal design under different conditions and the corresponding power allocation. Other benchmark designs have also been included for comparison.

In the simulation, the BS-Tx, the BS-Rx, and the CU are each equipped with a uniform linear array (ULA) with half wavelength spacing between consecutive antennas. We consider Rician fading for the communication channel, i.e., $\bm{H}_c = \sqrt{\frac{K_c}{K_c+1}} \bm{H}_c^{\text{los}} + \sqrt{\frac{1}{K_c+1}} \bm{H}_c^w$, where $\bm{H}_c^w$ is a CSCG random matrix with zero mean and unit variance, and $\bm{H}_c^{\text{los}} = \bm{a}_r^c(\theta_r) {\bm{a}_t^c}^T(\theta_t)$. Here, $\bm{a}_r^c(\theta_r^c)$ and $\bm{a}_t^c(\theta_t^c)$ denote the steering vectors at the CU receiver and the BS-Tx, and $\theta_r^c = \theta_t^c = \frac{\pi}{6}$ denote the AoA at the CU and the AoD at the BS-Tx, respectively. Furthermore, the noise power $\sigma_c^2$ at the CU and $\sigma_{s}^2$ at the BS-Rx are both normalized to be unity, the length of symbols in CPI is $L = 200$, and the number of antennas at the BS-Rx is $N_s = 12$.
%In the rank deficient case, we name the first $r$ subchannnels as communication subchannels and the other $M-r$ ones as sensing subchannels.
%
%\subsection{Benchmark Schemes}
%
%Besides the optimal design strategy that achieves the point on the Pareto boundary, we also propose several other practical designs, which are illustrated as follows.
%
%
%\subsection{Rank Deficient Case}

First, we consider the scenario where the number of transmit antennas at the BS-Tx is $M = 8$, the number of antennas at CU is $N_c = 6$, the Rician factor is $K_c=100$ and the power budget at the BS-Tx is $P = 800$ ($29.3$ dB). In this case, we have $r < M$, such that $\bm{Q}^*_c$ is rank-deficient and $\text{CRB}_C \rightarrow \infty$. Fig. \ref{fig:region_tlessthanM} shows the resultant C-R regions achieved by the optimal design and other benchmark schemes. It is observed that the C-R-region boundary by the optimal design outperforms those by the power splitting designs with equal power allocation and strongest eigenmode transmission. It is also observed that when the CRB is low, the three designs achieve similar C-R-region boundaries. Furthermore, as $\Gamma$ increases, the C-R-region boundary by the optimal design is observed to approach the capacity without sensing (i.e., $R_{\text{max}}$). This is consistent with the result in Remark \ref{remark:finite_SCRB}.

%will asymptotically get close to the dotted black one which shows the maximum capacity obtained without sensing, which is consistent with the previous analysis.

%Thus the resultant $\text{rank}(\bm{H}_c) = r = 6 < M$.
%The resultant $\mathcal{C}_{C-R}(P)$ together with the points achieved by other proposed sub-optimal designs within the region under the same $\bm{H}_c$ is given in Fig. .

\begin{figure}[htb]
	\centering
	\setlength{\abovecaptionskip}{-1mm}
	\setlength{\belowcaptionskip}{-4mm}
	\includegraphics[width=2.3in]{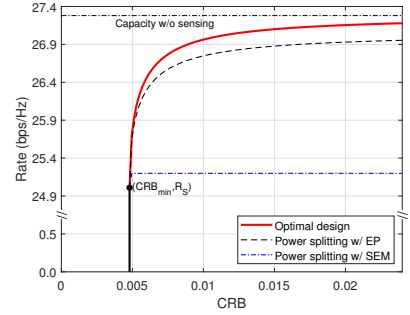}
	\caption{The C-R region in the case with $M = 8$ and $r = N_c = 6$. } \label{fig:region_tlessthanM}
	%	\textbf{[Check the template, the caption should start with Fig.]}
\end{figure}

\begin{figure}[htb]
	\centering
	\setlength{\abovecaptionskip}{-1mm}
	\setlength{\belowcaptionskip}{-4mm}
	\includegraphics[width=2.3in]{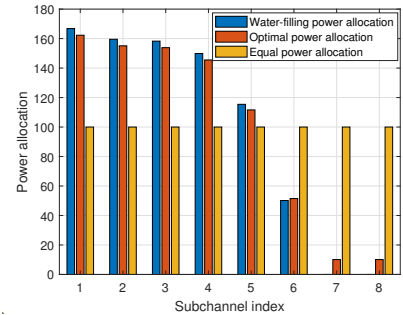}
	\caption{The power allocation in the case with $M = 8$ and $r = N_c = 6$, and $\Gamma = 0.0152$.} \label{fig:power_allocation_deficient}
	%	\textbf{[Check the template, the caption should start with Fig.]}
\end{figure}

\begin{figure}[htb]
	\centering
	\setlength{\abovecaptionskip}{-1mm}
	\setlength{\belowcaptionskip}{-4mm}
	%	\rule{12.8cm}{7.2cm}
	\includegraphics[width=2.3in]{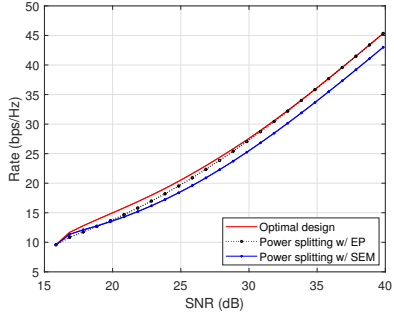}
	\caption{The rate versus the SNR in dB with $M = 8$ and $r = N_c = 6$, and $\Gamma = 0.1$. } \label{fig:rate_SNR} 
	%	\textbf{[Check the template, the caption should start with Fig.]}
\end{figure}

\begin{figure}[htb]
	\centering
	\setlength{\abovecaptionskip}{-1mm}
	\setlength{\belowcaptionskip}{-4mm}
	\includegraphics[width=2.3in]{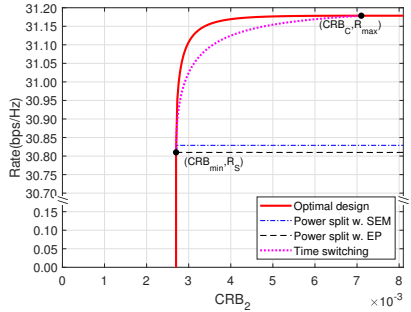}
	\caption{The C-R region in the case with $r = M = N_c = 6$.} \label{fig:C-R-region}
	%	\textbf{[Check the template, the caption should start with Fig.]}
\end{figure}

%Without loss of generality, we examine the power allocation of the optimal strategy at $\Gamma = 0.0152$ in Fig. \ref{fig:region_tlessthanM} and compare it with conventional water-filling and equal power allocation in

Fig. \ref{fig:power_allocation_deficient} shows the optimal power allocation in the case with $M = 8$, $N_c = 6$, and CRB threshold $\Gamma = 0.0152$, as compared to the water-filling and equal power allocations for rate maximization and CRB minimization, respectively. It is observed that the proposed optimal power allocations over the first six communication subchannels are monotonically non-increasing, which are higher than the constant power allocated to the next two sensing subchannels. This is consistent with Proposition \ref{lemma:Lemma_order_p}. It is also observed that the proposed optimal power allocations over the first five communication subchannels are lower than the corresponding water-filling power allocations, as more power should be allocated to other subchannels for facilitating the sensing. By contrast, the proposed optimal power allocations over the last three subchannels are higher than the corresponding water-filling power allocations, in order to meet the sensing requirement.

%, where the subchannels have been placed in decreasing order according to their subchannel gains from left to right. We can see that the power allocation towards the worst communication subchannel is larger than the power allocation towards those sensing subchannels that do not contribute to the sum rate. Besides, the power allocation is also monotonically decreasing from left to right. Both these two observations are

Fig. \ref{fig:rate_SNR} shows the rate versus the signal-to-noise ratio (SNR) (or equivalently the transmit power $P$) in the case with $M = 8$, $N_c = 6$, and  $\Gamma = 0.1$. It is observed that the optimal design performs best over the whole SNR regime. In the high SNR regime, the rate achieved by the power splitting with equal power allocation is observed to approach that by the optimal design. This can be explained based on Proposition \ref{pro:P_infinite}. In the low SNR, the power splitting with strongest eigenmode transmission is observed to approach the optimal design. 

%We can see that when $P$ is small, ''Power splitting w/ SEM'' design approaches the optimal design. On the other hand, when $P$ is large, ''Power splitting w/ EP'' design approaches the optimal design.

%\subsection{Full Column Rank Case}

Next, we consider another scenario where $M = N_c = 6$, $K_c=20$, and $P = 800$. In this case, we have $r = 6$, and $\bm{Q}_c^*$ is of full rank (as $P > P_0$ in  Remark \ref{remark:finite_SCRB}), such that $\text{CRB}_C$ is finite. Fig. \ref{fig:C-R-region} shows the resultant C-R regions. It is observed that the boundary point $(\text{CRB}_C, R_{\text{max}})$ exists and the C-R-region boundary achieved by the optimal design outperforms other benchmark schemes. The C-R-region boundary by time switching is observed to outperform the other two power splitting designs when the CRB value becomes large.

\section{Conclusion}
This paper investigated the fundamental performance tradeoff between the estimation CRB and the communication data rate in a point-to-point MIMO ISAC system with an extended radar target. We characterized the complete Pareto boundary of the resultant C-R region, by proposing the semi-closed-form optimal transmit covariance solution to a new CRB-constrained rate maximization problem. Numerical results were provided to show the C-R-region boundary achieved by the optimal design as compared to other benchmark schemes. We hope that this paper can provide insights on revealing the fundamental limits of MIMO ISAC. 
%How to extend this work to other setups with different communication and sensing models is interesting research directions for future work. 

%The achievable rate and the CRB have been used as the performance metrics in the characterization for communication and radar sensing, respectively. We give the explicit closed form expression for the two corner points in the C-R region under different conditions and formulate a problem to achieve the Pareto boundary between the two corners. We strictly prove that the problem can be simplified as the power allocation problem into each subchannel and derive the semi-closed form solutions to the problem. It is revealed that the optimal power allocation strategy unifies the sensing-favoured equal power allocation and the communication-favoured water-filling power allocation. Numerical results have been conducted to verify the design with corresponding insights and analysis.
%It is our hope that this work could provide some guidelines for the explicit design of the optimal unified waveform for both sensing and communication.

	%%%%%%%%%%%%%%%%%%%%
	\bibliographystyle{ieeetran}
	%\bibliographystyle{IEEEtran}
	
	%	\bibliographystyle{IEEEbib}
	%% Put references in BibTeX format in refs.bib.
	%% since for some reference, we need the abbreviation of the conference/journal/Transaction, first check it in BibGuru, then just change the full name of the conference/journal/Transaction into its corresponding abbreviation.
	\bibliography{refsv2}
	%%%%%%%%%%%%%%%%%%%%

\appendix

%\subsection{Proof of Proposition 1}\label{Proof:diagonal_optimal}
%Suppose the optimal solution $\tilde{\bm{Q}}^*$ is not a diagonal positive definite matrix,  we delete all its non-diagonal elements and get the resultant matrix denoted by $ \tilde{\bm{Q}}^* \circ \mv{I}$. By Hadamard inequality \cite{horn2012matrix}, we then have
%\begin{align}
%	\nonumber
%	\log_2 \det (\mv{I}_{M} + \frac{1}{\sigma_c^2}  \mv{\Sigma}^2 \tilde{\bm{Q}}^* )
%	< & \log_2 \det (\mv{I}_{M} + \frac{1}{\sigma_c^2}  \mv{\Sigma}^2 (\tilde{\bm{Q}}^* \circ \mv{I}) ) \\
%	= & \sum_{i=1}^r \log_2 (1+\frac{\lambda_i^2 [\tilde{\bm{Q}}^*]_{i,i}}{\sigma_c^2})
%\end{align}
%$\tilde{\bm{Q}}^* \circ \mv{I}$ also meets the power constraint in (\ref{equ:P1_1_power}) as
%\begin{align}
%	\operatorname{tr} (\tilde{\bm{Q}}^* \circ \mv{I}) = \operatorname{tr} (\tilde{\bm{Q}}^*) \leq P.
%\end{align}
%Furthermore, we have \cite{ohno2004capacity}
%\begin{align}
%	\operatorname{tr} \{(\tilde{\bm{Q}}^* \circ \mv{I})^{-1}\} \leq \operatorname{tr} \{ \left(\tilde{\bm{Q}}^*\right)^{-1}\} \leq \tilde{\Gamma}
%\end{align}
%Thus, $\tilde{\bm{Q}}^* \circ \mv{I}$ is indeed another feasible solution and it increases the objective function value, this inherently results in a contradiction to the assumption that $\tilde{\bm{Q}}^*$ is the optimal solution. Thus, the optimal solution to (P1.1) must be a diagonal matrix with all positive diagonal elements, i.e., $\tilde{\bm{Q}} = \operatorname{diag}(\bm{p}) = \operatorname{diag}(p_1,...,p_M)$, where $p_i>0, \forall i$.

\subsection{Proof of Proposition 2}\label{Proof:prime_dual_relationship}

We prove this proposition via the Lagrange duality method. Let $\mu \geq 0$ and $v \geq 0$ denote the dual variables associated with constraints (\ref{equ:P_1P_CRB}) and (\ref{equ:P_1P_Power}), respectively. By denoting $\bm{p} \triangleq \left[p_1,...,p_M\right]^T$, the partial Lagrangian of (P1.2) is expressed as
\begin{align}
	\nonumber
	\mathcal{L}(\bm{p},\mu,v)  = \sum_{i=1}^r \log_2 (1+\frac{\lambda_i^2 p_i}{\sigma_c^2}) - & \mu (\sum_{i=1}^M \frac{1}{p_i} - \tilde{\Gamma}) - \\
	v(\sum_{i=1}^M p_i - P),
\end{align}
and the corresponding dual function is given by
\begin{align}\label{equ:dual_function}
	g(\mu,v) = \max_{\bm{p} \geq \bm{0}} \mathcal{L}(\bm{p},\mu,v).
\end{align}
Accordingly, the dual problem of (P1.2) is given by
\begin{align}\label{P_1_dual}
	\text{(D1.2)}: \min _{\mu \geq 0, v \geq 0} g(\mu,v).
\end{align}
Since problem (P1.2) is convex and satisfies the Slater's condition, the strong duality holds between problem (P1.2) and its dual problem (D1.2) \cite{boyd2004convex}. As a result, we can optimally solve problem (P1.2) by equivalently solving the dual problem (D1.2). In the following, we first solve problem (\ref{equ:dual_function}) to obtain the dual function $g(\mu,v)$ and then solve (D1.2) to obtain the optimal dual solution $\mu^{\text{opt}}$ and $v^{\text{opt}}$. 

First, consider problem (\ref{equ:dual_function}) with given $\mu \geq 0$ and $v \geq 0$, and suppose that its optimal solution is given by $\bm{p}^*$.  By setting the partial derivatives of $\mathcal{L}(\bm{p},\mu,v)$ with respect to $p_i$'s to be zero, we have
\begin{align}\label{equ:Lag_zero1}
	\frac{\partial \mathcal{L}}{\partial p_i^*} = \frac{1}{\text{ln}2} (\frac{1}{1+\frac{\lambda_i^2 p_i^*}{\sigma_c^2}}) \frac{\lambda_i^2}{\sigma_c^2} + (\frac{\mu}{(p_i^*)^2}) - v = 0, 
	\forall i \in \{1,..., r\},
\end{align}
and
\begin{align}\label{equ:Lag_zero2}
	\frac{\partial \mathcal{L}}{\partial p_i^*} = \frac{\mu}{p_i^{*2}} - v = 0, \forall i \in \{r+1,..., M\}.
\end{align}
Based on (\ref{equ:Lag_zero1}), (\ref{equ:Lag_zero2}), and Cardano's formula for solving a cubic equation, we have the optimal solution to problem (\ref{equ:dual_function}) as
			\begin{align}
	\nonumber
	p_i^{*} =& -t_1 + \sqrt[3]{-t_2+\sqrt{t_2^2+t_3^3}} +  \sqrt[3]{-t_2-\sqrt{t_2^2+t_3^3}}\\
	& \quad \quad\quad\quad\quad \quad \quad \quad\quad\quad\quad \quad \forall i \in \{1,\ldots, r\}, \label{equ:general_exp_normal} \\
	p_i^{*} =& \sqrt{\frac{\mu}{v}}, \quad \forall  i \in \{r+1,\ldots, M\}, \label{equ:vanish_exp_normal}
\end{align}
where
$$
t_1 = \frac{b_i}{3a}, t_2 = \frac{27a^2d_i-9ab_ic+2b_i^3}{54a^3}, t_3 = \frac{3ac-b_i^2}{9a^2},
$$
with $a = v, b_i = v \frac{\sigma_c^2}{\lambda_i^2} - \frac{1}{\text{ln2}}, c = -\mu$, and $d_i = -\mu \frac{\sigma_c^2}{\lambda_i^2}$.

Next, we solve the dual problem (D1.2) to find the optimal dual solution $(\mu^{\text{opt}},v^{\text{opt}})$. Notice that the dual problem (D1.2) is always convex but non-differentiable in general. As a result, we can use subgradient-based methods such as the ellipsoid method \cite{boyd2004convex} to find the optimal solution. Towards this end, we use the fact that the subgradient of $g(\mu,v)$ is given by
\begin{align} \label{equ:sub_gradient_exp}
	\partial g |_{(\mu,v)} = [ -(\sum_{i=1}^{M} \frac{1}{p_i^*} - \Gamma_t ), -(\sum_{i=1}^{M} p_i^* - P)]^T.
\end{align}
Therefore, the optimal dual solution $(\mu^{\text{opt}},v^{\text{opt}})$ to (D1.2) can be obtained. 
%Since for any other $(\tilde{\mu},\tilde{v})$, we have
%\begin{align}\label{equ:gradient}
%	\nonumber
%	& g(\tilde{\mu},\tilde{v}) \geq \mathcal{L}(\bm{p}^*,\tilde{\mu},\tilde{v}) \\
%	& = g(\mu,v) + [\tilde{\mu}-\mu,\tilde{v}-v] \left[\begin{array}{r}
%		-(\sum_{i=1}^M \frac{1}{p_i^*} - \Gamma_t)  \\
%		-(\sum_{i=1}^M p_i^* - P).
%	\end{array}\right]
%\end{align}

Finally, by substituting $(\mu^{\text{opt}},v^{\text{opt}})$ into the formulas in \eqref{equ:Lag_zero2}, (\ref{equ:general_exp_normal}), and (\ref{equ:vanish_exp_normal}), the optimal solution to (P1.2) is given in \eqref{equ:Lag_zero1_in_pro},  (\ref{equ:general_exp_optimal}), and (\ref{equ:vanish_exp}). This thus completes the proof.

\subsection{Proof of Proposition 3}\label{Proof:lemma_power_allocation}
Based on (\ref{equ:vanish_exp}), it is evident that $p_{r+1}^{\text{opt}} = ... = p_{M}^{\text{opt}}>0$. Therefore, to verify Proposition \ref{lemma:Lemma_order_p}, we only need to prove that $p_1^{\text{opt}} \ge p_2^{\text{opt}} \ge ...\ge p_r^{\text{opt}} \ge p_{r+1}^{\text{opt}}$. 

First, we prove $p_r^{\text{opt}} \geq p_{r+1}^{\text{opt}}$ via contradiction. Suppose that $p_r^{\text{opt}} < p_{r+1}^{\text{opt}} = \sqrt{\frac{\mu^{\text{opt}}}{v^{\text{opt}}}}$. Then we have
\begin{align}
	0 \stackrel{(a)}{\leq} \frac{1}{\text{ln}2} (\frac{1}{1+\frac{\lambda_r^2 p_r^{\text{opt}}}{\sigma_c^2}}) \frac{\lambda_r^2}{\sigma_c^2} \stackrel{(b)}{=} \mu^{\text{opt}}(-\frac{1}{(p_r^{\text{opt}})^2}) + v^{\text{opt}} \stackrel{(c)}{<} 0,
\end{align}
where (a) follows from $p_r^{\text{opt}} > 0$, (b) is obtained based on (\ref{equ:Lag_zero1_in_pro}), and (c) holds based on 
the above presumption. This incurs a contradiction. Thus, we have $p_r^{\text{opt}} \geq p_{r+1}^{\text{opt}}$.

Next, we prove $p_i^{\text{opt}} \ge p_{i+1}^{\text{opt}}$, for any $i\in \{1,..., r-1\}$, by contradiction. Assume that  $p_i^{\text{opt}} < p_{i+1}^{\text{opt}}$. Then, based on (\ref{equ:Lag_zero1_in_pro}), we have
\begin{align}\label{equ:assump_1}
	\nonumber
	\frac{1}{\text{ln}2} (\frac{1}{\frac{1}{\lambda_i^2}+\frac{ p_i^{\text{opt}}}{\sigma_c^2}}) \frac{1}{\sigma_c^2}  & = v^{\text{opt}} - \mu^{\text{opt}} \frac{1}{(p_i^{\text{opt}})^2} < v^{\text{opt}} - \mu^{\text{opt}} \frac{1}{(p_{i+1}^{\text{opt}})^2} \\
	& = \frac{1}{\text{ln}2} (\frac{1}{\frac{1}{\lambda_{i+1}^2}+\frac{ p_{i+1}^{\text{opt}}}{\sigma_c^2}}) \frac{1}{\sigma_c^2}.
\end{align}
Furthermore, based on the presumption $\frac{p_{i}^{\text{opt}}}{\sigma_c^2} < \frac{p_{i+1}^{\text{opt}}}{\sigma_c^2}$ together with the fact that $\frac{1}{\lambda_{i}^2} \leq \frac{1}{\lambda_{i+1}^2}$, we have
\begin{align}
	\frac{1}{\lambda_i^2}+\frac{p_i^{\text{opt}}}{\sigma_c^2} < \frac{1}{\lambda_{i+1}^2}+\frac{p_{i+1}^{\text{opt}}}{\sigma_c^2},
\end{align}
which further yields
\begin{align}\label{eqn:app:2}
	1/(\frac{1}{\lambda_i^2}+\frac{p_i^{\text{opt}}}{\sigma_c^2}) > 1/(\frac{1}{\lambda_{i+1}^2}+\frac{p_{i+1}^{\text{opt}}}{\sigma_c^2}).
\end{align}
It is clear that \eqref{eqn:app:2} contradicts (\ref{equ:assump_1}). Thus, we must have $p_i^{\text{opt}} \geq p_{i+1}^{\text{opt}}$. Combining the above results finishes the proof.

\subsection{Proof of Proposition \ref{pro:P_infinite}} \label{Proof:Pro_P_infinite}

First, we consider that $r = M$. In this case, when $P \rightarrow \infty$, the optimal water-filling power allocation that maximizes the sum rate in \eqref{equ:P_1P} subject to the sum power constraint in \eqref{equ:P_1P_Power} reduces to the equal power allocation $p_i =P/M, \forall i\in \{1,\ldots,M\}$. Furthermore, such power allocation is shown to minimize the estimation CRB $\sum_{i=1}^M \frac{1}{p_i}$ in constraint \eqref{equ:P_1P_CRB}. Therefore, it follows that $p_i^{\text{opt}} =P/M, \forall i\in \{1,\ldots,M\}$. 

Next, we consider that $r<M$. Based on Propositions \ref{pro:prime_dual_relationship} and \ref{lemma:Lemma_order_p}, we have $p_{r+1}^{\text{opt}} = \ldots = p_{M}^{\text{opt}}$ in this case. Therefore, without loss of optimality, we use $p_s = p_{r+1} = \ldots = p_{M}$ to denote the transmit power allocated to sensing subchannels. Accordingly, problem (P1.2) is equivalently reformulated as
\begin{subequations}\label{equ:P_12_eq}
	\begin{align}
	\text{ } \max _{\{p_i \ge 0\}_{i=1}^r, p_s \ge 0} & \sum_{i=1}^r \log_2 (1+\frac{\lambda_i^2 p_i}{\sigma_c^2}) \\
		\label{equ:P_1P_CRB_eq}
		\text { s.t. } & \sum_{i=1}^r \frac{1}{p_i} + \frac{M-r}{p_s} \leq  \tilde{\Gamma} \\
		\label{equ:P_1P_Power_eq}
		& \sum_{i=1}^r p_i + (M-r) p_s \leq P.
		%		&\mathbf{Q} \succeq \mv{0}
	\end{align}
\end{subequations}
It follows from \eqref{equ:P_1P_CRB_eq} that $p_s \ge \frac{M-r}{\tilde{\Gamma}}.$
By setting
\begin{align}
\label{eq:ps} 
p_s = \frac{M-r}{\tilde{\Gamma}} 
\end{align}
 and dropping the estimation CRB constraint in \eqref{equ:P_1P_CRB_eq}, problem \eqref{equ:P_12_eq} is reduced to the following rate maximization problem:  
\begin{subequations}\label{equ:P_12_reduced}
	\begin{align}
		 \text{ } \max _{\{p_i \ge 0\}_{i=1}^r} & \sum_{i=1}^r \log_2 (1+\frac{\lambda_i^2 p_i}{\sigma_c^2}) \\ 
		\label{equ:P_1P_Power_eq_reduced}
	\text { s.t. } 	& \sum_{i=1}^r p_i \leq P - \frac{(M-r)^2}{\tilde{\Gamma}},
		%		&\mathbf{Q} \succeq \mv{0}
	\end{align}
\end{subequations}
for which the optimal value serves as an upper bound of that by \eqref{equ:P_12_eq}. As $P - \frac{(M-r)^2}{\tilde{\Gamma}} \rightarrow \infty$, it is clear that the equal power allocation, given by
\begin{align}\label{eqn:pr}
p_i = \frac{1}{r} (P-\frac{(M-r)^2}{\tilde{\Gamma}}), \forall i\in\{1,\ldots, r\},
\end{align}
 is optimal for problem \eqref{equ:P_12_reduced}. With $P\rightarrow \infty$, it can be shown that $\{p_i\}$ given in \eqref{eq:ps} and \eqref{eqn:pr} is feasible for problem (\ref{equ:P_12_eq}) and achieves the same value as the optimal value of problem \eqref{equ:P_12_reduced}. As a result, such power allocation is optimal for (\ref{equ:P_12_eq}) and thus (P1.2). This thus completes the proof.

%Notice that the optimal solution to this problem always uses up all the power budget $P$. Thus the constraint in (\ref{equ:P_1P_Power_eq}) is always tight. Furthermore, when $M > r$ and $P \to \infty$, given $\tilde{\Gamma}$, since the sensing subchannels do not contribute to the rate, $p_s$ will be finite such that the CRB constraint in (\ref{equ:P_1P_CRB_eq}) can be met while the remaining power will all be allocated into communication subchannels, i.e., $p_s = \frac{M-r}{\tilde{\Gamma}}$. The problem can then be reduced into 

%The solution to this problem is well known to be equal power allocation when $P \to \infty$ \cite{telatar1999capacity}, i.e., $p_i^{\text{opt}} = \frac{1}{r}(P-\frac{(M-r)^2}{\tilde{\Gamma}}), \forall i \in \{1,...,r\}$. Combining all the results yields the proof. 

%$p_s$ will be finite 

%it is intuitively to see that given $\tilde{\Gamma}$, the dedicated sensing subchannels will be allocated with equal power to meet the CRB constraint according to Proposition \ref{lemma:Lemma_order_p}. Given fixed $\tilde{\Gamma}$, when $P \to \infty$, the $M-r$ sensing subchannels will each be allocated with $\frac{M-r}{\tilde{\Gamma}}$ so that $\sum_{i=r+1}^M \frac{1}{p_i^{\text{opt}}} = \tilde{\Gamma}$ while the rest of the power will be equally allocated to the communication subchannels to increase the data rate. (\ref{equ:P_infinite}) can thus be obtained.

\end{document}